\documentclass[a4paper,UKenglish]{lipics}

\usepackage{microtype}
\usepackage{amsmath,amssymb}

\def\ket#1{|#1\rangle}
\def\F{{\mathbb F}}
\def\bm#1{\mathchoice{\mbox{\boldmath{$\displaystyle #1$}}}%
{\mbox{\boldmath{$\textstyle #1$}}}%
{\mbox{\boldmath{$\scriptstyle #1$}}}%
{\mbox{\boldmath{$\scriptscriptstyle #1$}}}}
\def\trace{\mathop{\rm tr}}

\def\figstar{{\unitlength0.08em\begin{picture}(10,10)\put(0,5){\line(1,0){10}}\put(5,0){\line(0,1){10}}\end{picture}}}
\def\figbox{{\unitlength0.08em\begin{picture}(10,10)\put(0,0){\line(1,0){10}}\put(0,10){\line(1,0){10}}\put(10,0){\line(0,1){10}}\put(0,0){\line(0,1){10}}\end{picture}}}

\usepackage{color}

\title{Symmetries of \penalty-1000 Codeword Stabilized Quantum
  Codes\footnote{This work was partially supported by NSERC, CIFAR,
    and IARPA.}}
\titlerunning{Symmetries of CWS Codes} 

\author[1]{Salman Beigi}
\author[2,3]{Jianxin Chen}
\author[4]{Markus Grassl}
\author[3]{Zhengfeng Ji}
\author[5]{Qiang Wang}
\author[2,3]{Bei Zeng}
\affil[1]{School of Mathematics, Institute for Research in Fundamental Sciences (IPM)\\
  Niavaran Square, Tehran, Iran\\
  \texttt{salman.beigi@gmail.com}}
\affil[2]{Department of Mathematics \& Statistics, University of Guelph\\
  50 Stone Road East, Guelph, Ontario, Canada\\
  \texttt{\{chenkenshin,zengbei\}@gmail.com}}
\affil[3]{Institute for Quantum Computing\\
  200 University Avenue West, Waterloo, Ontario, Canada\\
  \texttt{jizhengfeng@gmail.com}}
\affil[4]{Centre for Quantum Technologies, National University of Singapore\\
3 Science Drive 2, Singapore 117543\\
  \texttt{Markus.Grassl@nus.edu.sg}}
\affil[5]{School of Mathematics and Statistics, Carleton University\\
  1125 Colonel By Drive, Ottawa, Ontario, Canada\\
  \texttt{wang@math.carleton.ca}}
\authorrunning{S. Beigi, J. Chen, M. Grassl, Z. Ji, Q. Wang \& B. Zeng} 

\Copyright[by]{S. Beigi, J. Chen, M. Grassl, Z. Ji, Q. Wang \& B. Zeng}


\serieslogo{}
\EventShortName{TQC 2013}
\DOI{10.4230/LIPIcs.xxx.yyy.p}

\begin{document}

\maketitle

\begin{abstract}
Symmetry is at the heart of coding theory. Codes with symmetry,
especially cyclic codes, play an essential role in both theory and
practical applications of classical error-correcting codes.  Here we
examine symmetry properties for codeword stabilized (CWS) quantum
codes, which is the most general framework for constructing quantum
error-correcting codes known to date.  A CWS code $\mathcal{Q}$ can be
represented by a self-dual additive code $\mathcal{S}$ and a classical
code $\mathcal{C}$, i.\,e., $\mathcal{Q}=(\mathcal{S},\mathcal{C})$,
however this representation is in general not unique.  We show that
for any CWS code $\mathcal{Q}$ with certain permutation symmetry, one
can always find a self-dual additive code $\mathcal{S}$ with the same
permutation symmetry as $\mathcal{Q}$ such that
$\mathcal{Q}=(\mathcal{S},\mathcal{C})$.  As many good CWS codes have
been found by starting from a chosen $\mathcal{S}$, this ensures that
when trying to find CWS codes with certain permutation symmetry, the
choice of $\mathcal{S}$ with the same symmetry will suffice.  A key
step for this result is a new canonical representation for CWS codes,
which is given in terms of a unique decomposition as union stabilizer
codes.  For CWS codes, so far mainly the standard form
$(\mathcal{G},\mathcal{C})$ has been considered, where $\mathcal{G}$
is a graph state.  We analyze the symmetry of the corresponding graph
of $\mathcal{G}$, which in general cannot possess the same permutation
symmetry as $\mathcal{Q}$.  We show that it is indeed the case for the
toric code on a square lattice with translational symmetry, even if
its encoding graph can be chosen to be translational invariant.
 \end{abstract}

\section{Introduction}

Coding theory is an important component of information theory having a
long history dating back to Shannon's seminal 1948 paper that laid the
ground for information theory \cite{shannon-1948}.  Coding theory is
at the heart of reliable communication, where codes with symmetry,
especially cyclic codes, such as the Reed-Solomon codes, are among the
most widely used codes in practice \cite{MacWilliams:77}.

In recent years, it has become evident that quantum communication and
computation offer the possibility of secure and high rate information
transmission, fast computational solution of certain important
problems, and efficient physical simulation of quantum phenomena.
However, quantum information processing depends on the identification
of suitable quantum error-correcting codes (QECC) to make such
processes and machines robust against faults due to decoherence,
ubiquitous in quantum systems.  Quantum coding theory has hence been
extensively developed during the past 15 years
\cite{CRSS97,thesis:gottesman,nielsenchuang}.

Codeword stabilized (CWS) quantum codes are by far the most general
construction of QECC \cite{CSS+09}.  A CWS code $\mathcal{Q}$ can be
represented by a stabilizer state (i.\,e. a self-dual additive code)
$\mathcal{S}$ and a classical code $\mathcal{C}$,
i.\,e. $\mathcal{Q}=(\mathcal{S},\mathcal{C})$.  When $\mathcal{C}$ is
a linear code, the corresponding CWS code $\mathcal{Q}$ is actually a
stabilizer code.  Also, any CWS code is local Clifford
equivalent to a standard form $(\mathcal{G},\mathcal{C})$, where
$\mathcal{G}$ is a graph state \cite{CSS+09}.

The CWS construction encompasses stabilizer (additive) codes and all
the known non-additive codes with good parameters.  It also leads to
many new codes with good parameters, or good algebraic/combinatorial
properties, through both analytical and numerical methods.
Alternative perspectives of CWS codes have also been analyzed,
including the union stabilizer codes (USt) method
\cite{grassl-roetteler-2008a,grassl-roetteler-2008}, and the codes
based on graphs \cite{looi-code,yu-graphical-2007}. Concatenated codes
and their generalizations using CWS codes have been developed
\cite{BCG+11}, and decoding methods for CWS codes have been studied as
well \cite{LDP09}.

Given all the evidence that the CWS framework is a powerful method to
construct and analyze QECC, it remains unclear to what extent the
stabilizer state $\mathcal{S}$ and the classical code $\mathcal{C}$
can represent the symmetry of the CWS code
$\mathcal{Q}=(\mathcal{S},\mathcal{C})$ in general.  Given the vital
importance that the code symmetry plays in coding theory, this
understanding becomes crucial since if such a correspondence exists, it
can provide practical methods for constructing CWS codes with desired
symmetry from $\mathcal{S}$ and/or $\mathcal{C}$ with corresponding
symmetry.

Unfortunately, there is no immediate clue what answer one
can hope for. First of all, the representation
$\mathcal{Q}=(\mathcal{S},\mathcal{C})$ is not unique.  So for a given
CWS code $\mathcal{Q}$, there might be some stabilizer states
$\mathcal{S}$ and/or classical codes $\mathcal{C}$ which are more
symmetric than others.  Perhaps the best known example is the CWS
representation for the five-qubit code $\mathcal{Q}_5$, where in the
ideal case $\mathcal{S}$ can be chosen as a graph state corresponding
to the pentagon graph, and the $\mathcal{C}$ is chosen as the repetition
code $\{00000,11111\}$.  In this case, both $\mathcal{S}$ and
$\mathcal{C}$ nicely represent the cyclic symmetry of the five-qubit
code.

However, there are known `bad cases', too.  One example is the
seven-qubit Steane code $\mathcal{Q}_7$, where although the code
itself is cyclic, one cannot find any $\mathcal{S}$ corresponding to a
cyclic graph, even if local Clifford operations are
allowed~\cite{GKR02}.  Nonetheless, we know that the stabilizer group
for this code $\mathcal{Q}_7$ is invariant under cyclic shifts, and
the logical $Z$ operator can be chosen as $Z_L=Z^{\otimes 7}$,
therefore the logical $\ket{0}_L$ can be chosen as a cyclic stabilizer
code.  This is to say, there exists a representation for
$\mathcal{Q}_7=(\mathcal{S},\mathcal{C})$ such that $\mathcal{S}$ is
cyclic.  In general it remains unclear under which conditions a
representation for cyclic CWS code with a cyclic stabilizer state
$\mathcal{S}$ exists.

In this work, we address the symmetry properties of CWS codes.  We are
interested in the permutation symmetry of CWS codes, which includes
the important category of cyclic codes.  Our main question is, to
which extent can the representation $(\mathcal{S},\mathcal{C})$ and
the standard form $(\mathcal{G},\mathcal{C})$ reflect the symmetry
of the corresponding CWS code $\mathcal{Q}$.  We show that for any CWS
code $\mathcal{Q}$ with permutation symmetry, one can always find a
stabilizer state $\mathcal{S}$ with the same permutation symmetry as
$\mathcal{Q}$ such that $\mathcal{Q}=(\mathcal{S},\mathcal{C})$.  As
many good CWS codes are found by starting from a chosen $\mathcal{S}$,
this ensures that when trying to find CWS codes with certain
permutation symmetry, the choice of $\mathcal{S}$ with the same
symmetry will suffice.  A key step to reach this main result is to
obtain a canonical representation for CWS codes, which is in terms of
a unique decomposition as union stabilizer codes.

We know that for the standard form of CWS codes using graph states, it
is not always possible to find a graph with the same permutation
symmetry.  This is partially due to the fact that the local Clifford
operation transforming the CWS code into the standard form may break
the permutation symmetry of the original code.  Also, the graphs
usually can only represent the symmetry of the stabilizer generators
of the stabilizer state, but not the symmetry of the stabilizer state
in general.  We show that this is indeed the case for the toric code
on a two-dimensional square lattice with translational symmetry, even
if its encoding graph can be chosen to be translational invariant.

However, we show that the converse always holds, i.\,e., any graph
$\mathcal{G}$ and classical code $\mathcal{C}$ with certain
permutation symmetry yields a CWS code
$\mathcal{Q}=(\mathcal{G},\mathcal{C})$ with the same symmetry.

\section{Preliminaries}
The single-qudit (generalized) Pauli group is generated by the
operators $X$ and $Z$ acting on the qudit Hilbert space ${\mathbb
  C}^p$, satisfying $ZX=\omega XZ$, where
$\omega=\omega_p=\exp{2i\pi/p}$. For simplicity, throughout the paper,
we assume that $p$ is a prime, although our results naturally extend
to prime powers. Denote the computational basis of ${\mathbb C}^p$ by
$\{\ket{j}\colon j=0,1,\ldots,p-1\}$. Then, without loss of
generality, we can fix the operators $X$ and $Z$ such that
$X\ket{j}=\ket{j+1}$ and $Z\ket{j}=\omega^j\ket{j}$, respectively.
Let $I$ be the identity operator.
The set $\{X^aZ^b\colon a,b=0,\ldots,p-1\}$ of $p^2$ operators forms a
so-called nice unitary error basis which is a particular basis for the
vector space of $p\times p$ matrices \cite{Knill96b,Knill96a}.

The $n$-qudit Pauli group ${\mathcal P}_n$ consists of all local
operators of the form $\bm{M}=\alpha_{\bm{M}} M_1\otimes\dots\otimes
M_n$, where $\alpha_{\bm{M}}=\omega^k$ for some integer $k$ is an
overall phase factor, and $M_i=X_i^aZ_i^b$ for some
$a,b\in\{0,1,\ldots,p-1\}$, is an element of the single-qudit Pauli
group of qudit $i$.  We can write $\bm{M}$ as $\alpha_{\bm{M}} (M_1)_1
(M_2)_2\dots (M_n)_n$ or $\alpha_{\bm{M}} M_1 M_2\dots M_n$ when it is
clear what the qudit labels are.  The weight of an operator $\bm{M}$
is the number of tensor factors $M_i$ that differ from identity.

The $n$-qudit Clifford group ${\mathcal L}_n$ is the group of
$p^n\times p^n$ unitary matrices that map ${\mathcal P}_n$ to itself
under conjugation. The $n$-qudit local Clifford group is a subgroup in
${\mathcal L}_n$ containing elements of the form
$M_1\otimes\dots\otimes M_n$, where each $M_i$ is a single qudit
Clifford operation, i.\,e., $M_i\in{\mathcal L}_1$.

A stabilizer group $\mathcal{S}$ in the Pauli group ${\mathcal P}_n$ is
defined as an abelian subgroup of ${\mathcal P}_n$ which does not
contain $\omega I$.  A stabilizer consists of $p^{m}$ Pauli operators
for some $m\leq n$. As the operators in a stabilizer commute with each
other, they can be simultaneously diagonalized.  The common eigenspace
of eigenvalue $1$ is a stabilizer quantum code
$\mathcal{Q}=(\!(n,K,d)\!)_p$ with length $n$, dimension $K=p^{n-m}$,
and minimum distance $d$.  The projection $P_{\mathcal{Q}}$ onto the
code $\mathcal{Q}$ can be expressed as
\begin{equation}\label{eq:stab_code_projection}
P_{\mathcal{Q}}=\frac{1}{|\mathcal{S}|}\sum_{\bm{M}\in\mathcal{S}}\bm{M}.
\end{equation}
The centralizer $C(\mathcal{S})$ of the stabilizer $\mathcal{S}$ is
given by the elements in ${\mathcal P}_n$ which commute with all
elements in $\mathcal{S}$. For $m<n$, the minimum distance $d$ of the
code $\mathcal{Q}$ is the minimum weight of all elements in
$C(\mathcal{S})\setminus\mathcal{S}$.

If $m=n$, then there exists a unique $n$-qudit state $|\psi\rangle$
such that $\bm{M}|\psi\rangle = |\psi\rangle$ for every
$\bm{M}\in{\mathcal S}$. Such a state $|\psi\rangle$ is called a
stabilizer state, and the group ${\mathcal S}={\mathcal
  S}(|\psi\rangle)$ is called the stabilizer of $|\psi\rangle$.  A
stabilizer state can also be viewed as a self-dual code over the
finite field $\F_{p^{2}}$ under the trace inner product
\cite{Danielsen05}.  For a stabilizer state, the minimum distance is
defined as the minimum weight of the non-trivial elements in
$\mathcal{S}(|\psi\rangle)$~\cite{Danielsen05}.

A union stabilizer (USt) code of length $n$ is characterized by a
stabilizer code with stabilizer $\mathcal{S}=\langle
\bm{g}_1,\bm{g}_2,\ldots,\bm{g}_m\rangle$, where
$\bm{g}_1,\bm{g}_2,\ldots,\bm{g}_m$ are $m$ independent generators,
and a classical code $\mathcal{C}$ over $\F_p$ of length $m$.  Note
that for a given $\mathcal{S}$, the choice of the $m$ generators
$\bm{g}_j$ is not unique. Now for a classical code $\mathcal{C}$ of
length $m$ with $K$ codewords, for each codeword
$\mathbf{c}=(c_1,c_2,\ldots,c_m)\in\mathcal{C}$, the corresponding
quantum code is given by the subspace $V_{\bm c}$ stabilized by
$\omega^{c_1}\bm{g}_1$, $\omega^{c_2}\bm{g}_2$, \ldots,
$\omega^{c_m}\bm{g}_m$. Note that for ${\bm c}\neq
{\bm{c}}'\in\mathcal{C}$, the subspaces $V_{\bm c}$ and $V_{{\bm c}'}$
are mutually orthogonal. The corresponding USt code is then given by
the subspace $\bigoplus_{\bm c}V_{\bm c}$.

Therefore,
the combination of $\mathcal{S}$ (more precisely, the generators of
$\mathcal{S}$) and $\cal{C}$ gives an $(\!(n,2^{n-m}K)\!)_p$ USt quantum
code $\mathcal{Q}$.  Hence we denote a USt code $\cal{Q}$ by
$\cal{Q}=( \cal{S},\cal{C})$.  The projection onto $\mathcal{Q}$ can
be expressed as
\begin{equation}\label{eq:USt_projection}
P_{\mathcal{Q}}=\sum_{\bm{c}\in\mathcal{C}}\frac{1}{p^m}\sum_{\bm{y}\in\F_p^m}\omega^{\bm{c}\cdot\bm{y}}
\bm{g}_1^{y_1}\dots\bm{g}_m^{y_m},
\end{equation}
where we identify the elements $y_i$ of the finite field with integers modulo $p$.

A CWS code $\cal{Q}$ of length $n$ is a USt code with $m=n$. That is,
it is characterized by a stabilizer state with stabilizer
$\mathcal{S}$ and a classical code $\mathcal{C}$ of length $n$.  For a
CWS code $\cal{Q}$ given by $\cal{Q}=( \cal{S},\cal{C})$, the
stabilizer $\cal{S}$ always corresponds to a unique stabilizer
state. We will then refer to $\mathcal{S}$ as the stabilizer state
when no confusion arises.

For a CWS code, the projection $P_{\cal{Q}}$ onto the code space is
given by
\begin{equation}
\label{eq:CWSproj}
P_{\mathcal Q}
=\sum_{\bm{t}\in{\mathcal{C}}}\frac{1}{p^n}\sum_{\bm{x}\in\F_p^n} \omega^{\bm{t}\cdot\bm{x}}\bm{g}_1^{x_1}\dots\bm{g}_n^{x_n},
\end{equation}
where we again identify the elements $x_i$ of the finite field with
integers modulo $p$.

A CWS code has a permutation symmetry $\sigma$ if
\begin{equation}
P_{\mathcal Q}^{\sigma}=P_{\mathcal Q},
\end{equation}
where $P_{\mathcal Q}^{\sigma}$ is the projection onto the  space
obtained by permuting the qudits of the code $\mathcal{Q}$ according
to $\sigma$.


\section{Canonical form of CWS codes}
For a given a CWS code $\mathcal{Q}=(\mathcal{S},\mathcal{C})$, there
might exist another stabilizer state $\mathcal{S}'$ and another
classical code $\mathcal{C}'$ such that
$\mathcal{Q}=(\mathcal{S}',\mathcal{C}')$.  In other words, the
representation of a CWS code by the stabilizer state $\mathcal{S}$ and
the classical code $\mathcal{C}$ is non-unique.

In order to discuss the relationship between the symmetry of the CWS
code $\mathcal{Q}$ and that of the stabilizer state $\mathcal{S}$, we
first need to explore the relationship between the different
representations of $\mathcal{Q}$ (i.\,e., the relationship between
$\mathcal{S}$ and $\mathcal{S}'$, as well as the relationship between
$\mathcal{C}$ and $\mathcal{C}'$).

Let us start by recalling that a stabilizer code can be viewed as a
CWS code where the classical code is a linear code \cite{CSS+09}.  A
simple way to see this is that for a given stabilizer code
$\mathcal{Q}_s$ with stabilizer generated by $\mathcal{S}=\langle
\bm{g}_1,\bm{g}_2,\ldots,\bm{g}_m\rangle$, which is a code of
dimension $p^{n-m}$, we can choose the larger stabilizer
$\mathcal{S}'=\langle
\bm{g}_1,\bm{g}_2,\ldots,\bm{g}_m,\bar{Z}_1,\ldots,\bar{Z}_{n-m}\rangle$,
where $\bar{Z}_1,\ldots,\bar{Z}_{n-m}\in C(\mathcal{S})$ mutually
commute.  Now choose the classical code
$\mathcal{C}'=\bigl\{(0,\ldots,0,x_{m+1},\ldots,x_n)\colon
x_j\in\{0,\ldots,p-1\}\bigr\}$ of length $n$ with $p^{n-m}$ codewords,
where the first $m$ coordinates of each codeword are zero. Then we
have $\mathcal{Q}_s=(\mathcal{S}',\mathcal{C}')$, i.\,e., the
stabilizer code $\mathcal{Q}_s$ can then be viewed as a CWS code with
stabilizer state $\mathcal{S}'$ and classical code
$\mathcal{C}'$. However, note that the choice of $\mathcal{S}'$ (and
hence $\mathcal{C}'$) is non-unique, as in particular the choice of
$\bar{Z}_1,\ldots,\bar{Z}_{n-s}\in C(\mathcal{S})$ is non-unique.

\begin{example}
As an example, consider the five-qubit
code with stabilizer
\begin{equation}
\bm{g}_1=XZZXI,\quad
\bm{g}_2=IXZZX,\quad
\bm{g}_3=XIXZZ,\quad
\bm{g}_4=ZXIXZ.
\end{equation}
In the CWS picture, the stabilizer state can be chosen as
\begin{equation}
\mathcal{S}=\langle \bm{g}_1,\bm{g}_2,\bm{g}_3,\bm{g}_4,\bm{Z}_L\rangle,
\end{equation}
where $\bm{Z}_L=Z^{\otimes 5}$ is the logical $Z$ operator.
Alternatively, one can choose the stabilizer state
\begin{equation}
\mathcal{S}'=\langle \bm{g}_1,\bm{g}_2,\bm{g}_3,\bm{g}_4,\bm{X}_L\rangle,
\end{equation}
where $\bm{X}_L=X^{\otimes 5}$ is the logical $X$ operator.
For both $S$ and $S'$, the classical code can be chosen as
$\mathcal{C}=\{00000,00001\}$.
\end{example}

Similarly, a USt code $(\mathcal{S},\mathcal{C})$ can be viewed as a
CWS code $(\mathcal{S}',\mathcal{C}')$ with the classical code
$\mathcal{C}'$ of length $n$ possessing some coset structure, i.\,e.,
$\mathcal{C}'=\bigcup_{\bm{t}_i\in \tilde{\mathcal{C}}}
\mathcal{C}_0+\bm{t}_i$, where $\mathcal{C}_0$ is a linear code. This
linear code $\mathcal{C}_0$ of length $n$ can be readily chosen as the
classical code for the CWS representation of the stabilizer code
$\mathcal{S}$. The code $\tilde{\mathcal{C}}$ of length $n$ can be
derived from $\mathcal{C}$ of length $m$ by appending $n-m$ zero
coordinates.  However, again, the choices of $\mathcal{S}'$ and
$\mathcal{C}'$ are non-unique.


In the general situation, we have some freedom in choosing the
stabilizer state when representing a stabilizer code or a USt code in
the CWS framework.  Consequently, for a given CWS code $\mathcal{Q}$,
there are also many different ways to write it in terms of a USt code
in general.  We will show, however, that we can always obtain a unique
stabilizer $\mathcal{S}$, when expressing a given CWS code as a USt
code.  The following theorem gives a canonical form for any CWS code.

\begin{theorem}
\label{thm:unique}
Every CWS code has a unique representation as a union stabilizer code.
\end{theorem}


\begin{proof}
To prove this theorem, we will need some lemmas.

\begin{lemma}[translational invariant codes]
\label{lm:coset_union}
Let $\mathcal{C}\subset \F_p^n$ be a code over $\F_p$ with
$|\mathcal{C}|=M$ and assume that for some non-zero $\bm{s}\in\F_p^n$
we have $\mathcal{C}=\mathcal{C}+\bm{s}$, i.\,e., the code is
invariant with respect to translation by $\bm{s}$.  Then $\mathcal{C}$
can be written as a disjoint union of cosets of the one-dimensional
space $\mathcal{C}_0=\langle \bm{s}\rangle$ generated by $\bm{s}$,
i.\,e.,
\[
\mathcal{C}=\bigcup_{\bm{t}_i\in\mathcal{C}'} \mathcal{C}_0+\bm{t}_i,
\]
where $\mathcal{C}'\subset\F_p^n$ with $|\mathcal{C}'|=M/p$.
\end{lemma}
\begin{proof}
By assumption, for every $\bm{x}\in\mathcal{C}$, the vector
$\bm{x}+\bm{s}$ is in the code as well.  Hence we can arrange the
elements of $\mathcal{C}$ as follows:
\[
\begin{array}{r|ccccc}
\mathcal{C}'        &\bm{t}_1       &\bm{t}_2       &\ldots &\bm{t}_{M/p}\\
\hline
\mathcal{C}'+\bm{s} &\bm{t}_1+\bm{s}&\bm{t}_2+\bm{s}&\ldots &\bm{t}_{M/p}+\bm{s}\\
\mathcal{C}'+2\bm{s} &\bm{t}_1+2\bm{s}&\bm{t}_2+2\bm{s}&\ldots&\bm{t}_{M/p}+2\bm{s}\\
\multicolumn{1}{c|}{\vdots}&\vdots&\vdots&\ddots&\vdots\\
\mathcal{C}'+(p-1)\bm{s} &\bm{t}_1+(p-1)\bm{s}&\bm{t}_2+(p-1)\bm{s}&\ldots &\bm{t}_{M/p}+(p-1)\bm{s}
\end{array}
\]
Every column in this arrangements is a coset $\mathcal{C}_0+\bm{t}_i$.
\end{proof}

\begin{lemma}[vanishing character sum]
\label{lm:shift}
Let $\mathcal{C}\subset \F_q^n$ be an arbitrary code of length $n$.
Assume that the function
\[
f\colon \F_p^n \rightarrow \mathbb{C};\quad
f(\bm{y})=\sum_{\bm{c}\in\mathcal{C}}\omega^{\bm{c}\cdot\bm{y}},
\]
where $\omega=\exp(2\pi i/p)$, vanishes outside a proper subspace
$V_0<\F_p^n$.  Then there exists a non-zero vector $\bm{s}\in\F_p^n$
such that $\mathcal{C}=\mathcal{C}+\bm{s}$.  What is more, the code
$\mathcal{C}$ can be written as a union of cosets of the linear code
$\mathcal{C}_0=V_0^\bot$, i.\,e.,
\begin{equation}\label{eq:union_cosets}
\mathcal{C}=\bigcup_{\bm{t}\in\mathcal{C}'} \mathcal{C}_0+\bm{t}.
\end{equation}
\end{lemma}

\begin{proof}
Let $\chi_C(\bm{y})$ denote the characteristic function of the code
$\mathcal{C}$, i.\,e., $\chi_{\mathcal{C}}(\bm{y})\in\{0,1\}$, and
$\chi_{\mathcal{C}}(\bm{y})=1$ if and only if
$\bm{y}\in\mathcal{C}$. Define
$g(\bm{y})=1-(1-\omega)\chi_{\mathcal{C}}(\bm{y})$. Then
$\bm{g}(y)=\omega^{\chi_{\mathcal{C}}(\bm{y})}$.

The Fourier transform of $g(\bm{y})$ over $\F_p^n$ reads
\begin{alignat*}{5}
\hat{g}(\bm{y})
&=\frac{1}{\sqrt{p^n}}\sum_{\bm{x}\in\F_p^n} \omega^{\bm{x}\cdot\bm{y}}g(\bm{x})
\\
&=\frac{1}{\sqrt{p^n}}\sum_{\bm{x}\in\F_p^n} \omega^{\bm{x}\cdot\bm{y}}(1-(1-\omega)\chi_{\mathcal{C}}(\bm{x}))
=\sqrt{p^n}\delta_{\bm{y},\bm{0}}-\frac{1-\omega}{\sqrt{p^n}}\sum_{\bm{x}\in\F_p^n} \omega^{\bm{x}\cdot\bm{y}}\chi_{\mathcal{C}}(\bm{x})\\
&=\sqrt{p^n}\delta_{\bm{y},\bm{0}}-\frac{1-\omega}{\sqrt{p^n}}\sum_{\bm{c}\in C} \omega^{\bm{c}\cdot\bm{y}}
=\sqrt{p^n}\delta_{\bm{y},\bm{0}}-\frac{1-\omega}{\sqrt{p^n}}f(\bm{y}),
\end{alignat*}
where $\delta_{\bm{y},\bm{0}}=1$ if $\bm{y}=\bm{0}$, and $\delta_{\bm{y},\bm{0}}=0$
otherwise.

This shows that for $\bm{y}\ne\bm{0}$, the Fourier transform
$\hat{g}(\bm{y})$ is proportional to the function $f(\bm{y})$, and
hence $\hat{g}$ vanishes outside of $V_0$ as well.  Recall that that
$\dim V_0 \le n-1$, as $V_0$ is a proper subspace by assumption.
Let $\bm{s}\in V_0^\bot$ be a non-zero vector that is orthogonal to all
vectors in $V_0$. Furthermore, let $V_0^{c}=\F_p^n\setminus V_0$
denote the set-complement of $V_0$ in the full vector space.

We want to show that the code $\mathcal{C}$ is invariant with respect
to translations by $\bm{s}$, i.\,e., $\mathcal{C}=\mathcal{C}+\bm{s}$
or equivalently,
$\chi_{\mathcal{C}}(\bm{y}+\bm{s})=\chi_{\mathcal{C}}(\bm{y})$.  This
is in turn equivalent to showing that $g(\bm{y})=g(\bm{y}+\bm{s})$. In
the following, $\mathcal{F}^{-1}$ denotes the inverse Fourier
transform:
\begin{alignat*}{5}
g(\bm{y}+\bm{s})=(\mathcal{F}^{-1}\hat{g})(\bm{y}+\bm{s})
&=\frac{1}{\sqrt{p^n}}\sum_{\bm{x}\in\F_p^n} \omega^{-\bm{x}\cdot(\bm{y}+\bm{s})}\hat{g}(\bm{x})\\
&=\frac{1}{\sqrt{p^n}}\sum_{\bm{x}\in V_0} \omega^{-\bm{x}\cdot(\bm{y}+\bm{s})} \hat{g}(\bm{x})
+\frac{1}{\sqrt{p^n}}\sum_{\bm{x}\in V_0^c} \omega^{-\bm{x}\cdot(\bm{y}+\bm{s})} \hat{g}(\bm{x})\\
&=\frac{1}{\sqrt{p^n}}\sum_{\bm{x}\in V_0} \omega^{-\bm{x}\cdot\bm{s}}\omega^{-\bm{x}\cdot\bm{y}} \hat{g}(\bm{x})\displaybreak[3]\\
&=\frac{1}{\sqrt{p^n}}\sum_{\bm{x}\in V_0} \omega^{-\bm{x}\cdot\bm{y}} \hat{g}(\bm{x})\displaybreak[3]\\
&=\frac{1}{\sqrt{p^n}}\sum_{\bm{x}\in V_0} \omega^{-\bm{x}\cdot\bm{y}} \hat{g}(\bm{x})
+\frac{1}{\sqrt{p^n}}\sum_{\bm{x}\in V_0^c}\omega^{-\bm{x}\cdot\bm{y}} \hat{g}(\bm{x})\displaybreak[3]\\
&=\frac{1}{\sqrt{p^n}}\sum_{\bm{x}\in \F_p^n} \omega^{-\bm{x}\cdot\bm{y}} \hat{g}(\bm{x})\\
&=({\mathcal F}^{-1}{\hat{g}})(\bm{y})=g(\bm{y}).
\end{alignat*}
Here we have used the fact that $\hat{g}(\bm{x})$ vanishes outside of $V_0$
and that $\bm{s}$ is orthogonal to all vectors in $V_0$.

From Lemma~\ref{lm:coset_union}, it follows that the code
$\mathcal{C}$ can be written as a union of cosets of the code
$\mathcal{C}_0=V_0^\bot$ generated by all vectors $\bm{s}$ that are orthogonal
to $V_0$.
\end{proof}

Now we are ready to prove Theorem~\ref{thm:unique}.  Let
$P_{\mathcal{Q}}$ denote the projection operator onto a CWS code
$\mathcal{Q}=(\!(n,K,d)\!)_p$, i.\,e.
\begin{alignat}{6}\label{eq:CWS_projection}
P_{\mathcal{Q}}
&{}={}&\sum_{\bm{t}\in{\mathcal{C}}}\frac{1}{p^n}\sum_{\bm{x}\in\F_p^n} \omega^{\bm{t}\cdot\bm{x}}\bm{g}_1^{x_1}\dots\bm{g}_n^{x_n}
&&{}={}&\frac{1}{p^n}\sum_{\bm{x}\in\F_p^n}\left(\sum_{\bm{t}\in{\mathcal{C}}}\omega^{\bm{t}\cdot\bm{x}}\right)\bm{g}_1^{x_1}\dots\bm{g}_n^{x_n}\nonumber\\
&&&&{}={}&\frac{1}{p^n}\sum_{\bm{x}\in\F_p^n}\alpha_{\bm{x}}\bm{g}_1^{x_1}\dots\bm{g}_n^{x_n}
\end{alignat}
where $\bm{g}_1,\ldots,\bm{g}_n$ are the generators of the stabilizer,
and $\mathcal{C}=(n,K)_p$ is a classical code.

First note that the coefficients $\alpha_{\bm{x}}$ in
\eqref{eq:CWS_projection} are uniquely determined since the $p^n$
operators $\{\bm{g}_1^{x_1}\dots\bm{g}_n^{x_n}\colon
\bm{x}\in\F_p^n\}$ are a subset of the error-basis of linear operators
on the space $\mathbb{C}^{p^n}$.  The coefficient $\alpha_{\bm{x}}$ is
proportional to $\trace(\bm{g}_1^{x_1}\dots\bm{g}_n^{x_n}\cdot
P_{\mathcal Q})$. On the other hand,
$\alpha_{\bm{x}}=\sum_{\bm{t}\in{\mathcal{C}}}\omega^{\bm{t}\cdot\bm{x}}=f(\bm{x})$,
where $f(\bm{x})$ is the function appearing in Lemma~\ref{lm:shift}.
So if the coefficients $\alpha_{\bm{x}}=f(\bm{x})$ vanish outside of
a proper subspace $V_0<\F_p^n$, the classical code $\mathcal{C}$ can
be decomposed as union of cosets of $\mathcal{C}_0=V_0^\bot$.  Then
(\ref{eq:CWS_projection}) can be re-written as follows:
\begin{alignat}{6}
\label{eq:CWS_projection2}
P_{\mathcal{Q}}
&{}={}\frac{1}{p^n}\sum_{\bm{x}\in V_0}\left(\sum_{\bm{t'}\in{\mathcal{C}}'}\sum_{\bm{c}\in\mathcal{C}_0}\omega^{(\bm{t}'+\bm{c})\cdot\bm{x}}\right)\bm{g}_1^{x_1}\dots\bm{g}_n^{x_n}\displaybreak[3]\nonumber\\
&{}={}\frac{1}{p^n}\sum_{\bm{x}\in V_0}\left(\sum_{\bm{c}\in\mathcal{C}_0}\omega^{\bm{c}\cdot\bm{x}}\sum_{\bm{t'}\in{\mathcal{C}}'}\omega^{\bm{t}'\cdot\bm{x}}\right)\bm{g}_1^{x_1}\dots\bm{g}_n^{x_n}\displaybreak[3]\nonumber\\
&{}={}\frac{|\mathcal{C}_0|}{p^n}\sum_{\bm{x}\in V_0}\left(\sum_{\bm{t'}\in{\mathcal{C}}'}\omega^{\bm{t}'\cdot\bm{x}}\right)\bm{g}_1^{x_1}\dots\bm{g}_n^{x_n}
\end{alignat}
In the last step we have used the fact that the spaces $V_0$ and
$\mathcal{C}_0$ are orthogonal to each other, i.\,e., the inner
product $\bm{c}\cdot\bm{x}$ vanishes.  Now assume that the space $V_0$
has dimension $m$ and that $\{\bm{b}_1,\ldots,\bm{b}_m\}\subset\F_p^n$
is a basis of $V_0$.  Then every vector $\bm{x}\in V_0$ can be
expressed as $\bm{x}=\sum_{j=1}^m y_j\bm{b}_j$. For every
$\bm{t}'\in\mathcal{C}'$ we define the vectors $\bm{s}\in\F_p^m$ with
$s_j=\sum_{i=1}^n t_i b_{ji}$, forming another classical code
$\mathcal{D}\subset\F_p^m$.  Further, we define the $m$ operators
$\tilde{\bm{g}}_j=\prod_{i=1}^n\bm{g}_i^{b_{ji}}$.  This allows us to
express \eqref{eq:CWS_projection2} as
\begin{eqnarray}
\label{eq:CWS_projection3}
P_{\mathcal{Q}}
&=&\frac{1}{p^m}\sum_{\bm{y}\in\F_p^m}
  \left(\sum_{\bm{s}\in{\mathcal{D}}}\omega^{\bm{s}\cdot\bm{y}}\right)\tilde{\bm{g}}_1^{y_1}\dots\tilde{\bm{g}}_m^{y_m}.
\end{eqnarray}
Hence, whenever the classical code associated to a CWS code has some
non-trivial shift invariance, the projection onto a CWS code can be
expressed as a projection onto a USt code
(cf. \eqref{eq:USt_projection}), thereby increasing the dimension of
the underlying stabilizer code and reducing the size of the classical
code. In order to obtain a unique representation, we may assume that
the stabilizer code is of maximal dimension, and hence the classical
code is ``without any linear structure.''

In order to show uniqueness, consider the coefficients
$\trace(\bm{M}\cdot P_{\mathcal{Q}})$ of the expansion of the
projection $P_{\mathcal{Q}}$ in terms of the operator basis formed by
the $n$-qudit Pauli matrices $\bm{M}$.  Clearly, we have
$\{\bm{M}\colon \trace(\bm{M}\cdot P_{\mathcal{Q}})\ne 0\}\subset
\mathcal{S}=\langle \tilde{\bm{g}}_1,\ldots,\tilde{\bm{g}}_m\rangle$.
If the group $\mathcal{S}'=\langle\bm{M}\colon \trace(\bm{M}\cdot
P_{\mathcal{Q}})\ne 0\rangle$ was a proper subgroup of $\mathcal{S}$,
the coefficients
$\sum_{\bm{s}\in\mathcal{D}}\omega^{\bm{s}\cdot\bm{y}}$ would vanish
for $\bm{y}$ outside a proper subspace $V_0<\F_p^m$, contradicting
the assumption the classical code $\mathcal{D}$ has no linear
structure.  

Note that the stabilizer $\mathcal{S}$ is only unique up to the choice
of some phase factors of the error basis.  For example, replacing
$\tilde{\bm{g}}_1$ by $\omega{\tilde{\bm{g}}}_1$ will introduce some
phase factor which has to be compensated by changing the first
coordinate $s_1$ of the codewords $\bm{s}$ of the classical code
$\mathcal{D}$. To finally fix the degree of freedom, we can enforce
${\bm g}_i= M_1\otimes\cdots\otimes M_n$, with $M_j=X_j^{a}Z_j^{b}$
for $j=1,2,\ldots,n$ and $a,b\in\{0,1,\ldots p-1\}$.
\end{proof}

\section{Symmetries of the stabilizer state of a CWS code}

We are now ready to discuss the relationship between the symmetries of
the CWS code $\mathcal{Q}$ and that of the corresponding stabilizer state
$\mathcal{S}$.

\begin{theorem}
\label{thm:CWSsym}
For any CWS code $\mathcal{Q}$ with permutation symmetry $\sigma$, there exists a
stabilizer state $\mathcal{S}$ with the same permutation symmetry $\sigma$ such that
$\mathcal{Q}=(\mathcal{S}, \mathcal{C})$.
\end{theorem}

\begin{proof} To prove this theorem, we will need some lemmas.
\begin{lemma}
\label{lm:proj}
If the projection operator $P_{\mathcal C}$ given in
Eq.~\eqref{eq:CWS_projection} is invariant under a permutation
$\sigma$ of the qudits, then the stabilizer code related to expressing
$P_{\mathcal{C}}$ in terms of a USt code as in
Eq.~\eqref{eq:CWS_projection3} is invariant with respect to the
permutation as well.
\end{lemma}
\begin{proof}
The statement follows directly from the uniqueness of the stabilizer
group
$\mathcal{S}=\langle\tilde{\bm{g}}_1,\ldots,\tilde{\bm{g}}_m\rangle$
generated by the operators in Eq.~\eqref{eq:CWS_projection3}.
\end{proof}

We now prove a lemma for a special case of Theorem~\ref{thm:CWSsym}, when
the CWS code is a Calderbank–Shor–Steane (CSS) code \cite{CS96,steane96}.

\begin{lemma}
\label{lm:CSSsym}
For a CSS code $\mathcal{Q}$ with permutation symmetry $\sigma$, there
exists a stabilizer state $\ket{\psi}\in\mathcal{Q}$ such that
$\ket{\psi}$ has the same permutation symmetry as $\mathcal{Q}$.
\end{lemma}

\begin{proof}
For a CSS code $\mathcal{Q}$, the stabilizer generators can always be
chosen such that every generator is either a tensor product of powers
of $X$ (denoted by $\mathcal{S}_X$) or a tensor product of powers of
$Z$ (denoted by $\mathcal{S}_Z$). We can use the following matrix
form:
\[ 
\left[
\begin{array}{c|c}
\mathcal{S}_X & 0 \\
0 & \mathcal{S}_Z \\
\end{array}
\right]
\] 
As the permutation symmetry $\sigma$ of $\mathcal{Q}$ does not change
the type of an operator, both $\mathcal{S}_X$ and $\mathcal{S}_Z$ have
necessarily the same symmetry $\sigma$. Furthermore, the logical
operators can also be chosen as either tensor products of powers of
$X$ or tensor products of powers of $Z$, which correspond to the dual
of the classical codes associated to either the $Z$ stabilizers or the
$X$ stabilizers, respectively.  Without loss of generality let us
choose a set $\mathcal{L}_Z$ of commuting logical operators which are
all of $Z$ type.  Then group generated by the set
$\mathcal{S}_X\cup\mathcal{S}_Z\cup\mathcal{L}_Z$ of mutually
commuting operators is again invariant under the permutation $\sigma$.
As the stabilizer group is maximal, it stabilizes a unique state
$\ket{\psi}$.  Hence $\ket{\psi}$ is the stabilizer state with the
desired symmetry, and the CSS code can be expressed as CWS code in
terms of $\ket{\psi}$ and some classical code $\mathcal{C}$.
\end{proof}

We now prove a lemma for the stabilizer code case of Theorem~\ref{thm:CWSsym}, 
which improves the result of Lemma~\ref{lm:CSSsym}.

\begin{lemma}
\label{lm:Stabilizersym}
For a stabilizer code $\mathcal{Q}$ with permutation symmetry $\sigma$,
there exists a stabilizer state
$\ket{\psi}\in\mathcal{Q}$ such that $\ket{\psi}$ has
the same permutation symmetry as $\mathcal{Q}$.
\end{lemma}

\begin{proof}
To prove this lemma, we shall use a standard form for
stabilizers (see \cite[Section 10.5.7]{nielsenchuang}):
\[ 
\left[
\begin{array}{ccc|ccc}
I & A_1 & A_2 & B & 0 & C \\
0 & 0 & 0 & D & I & E \\
\end{array}
\right]
=\left[
\begin{array}{c|c}
\mathcal{S}_X & \mathcal{S}_Z \\
0 & \mathcal{S}'_Z \\
\end{array}
\right]
=\left[
\begin{array}{c}
\mathcal{S} \\
\mathcal{S}' \\
\end{array}
\right]
\] 
where $A_1$ is an $r\times(n-k-r)$ matrix, $A_2$ is an $r\times k$
matrix, $B$ is an $r\times r$ matrix, $C$ is an $r\times k$ matrix,
$D$ is an $(n-k-r)\times r$ matrix, and $E$ is an $(n-k-r)\times k$
matrix.  Similar as in the CSS case, we can choose a set
$\mathcal{L}_Z$ of commuting logical operators which are all of $Z$
type. In matrix form, they are given by $[0\;0\;0|-A_2^t\;0\;I]$.
Then the group generated by the mutually commuting operators in
$\mathcal{S}\cup\mathcal{S}'\cup\mathcal{L}_X$ stabilizes a unique
state $\ket{\psi}$ which is invariant with respect to the permutation
$\sigma$.  Hence $\ket{\psi}$ is the stabilizer state with the desired
symmetry that can be used to express $\mathcal{Q}$ as CWS code with
some classical code $\mathcal{C}$.
\end{proof}
To prove Theorem~\ref{thm:CWSsym}, given a CWS code $\mathcal{Q}$, we
first find its unique decomposition as a USt code
$\mathcal{Q}=(\mathcal{S},\mathcal{C})$, based on
Theorem~\ref{thm:unique}. Here $\mathcal{S}$ is in general a stabilizer
code with $m=n-k$ generators. If $\mathcal{Q}$ has a permutation
symmetry $\sigma$, then according to Lemma~\ref{lm:proj}, the
stabilizer code $\mathcal{S}$ must also have the symmetry $\sigma$.
Now according to Lemma~\ref{lm:Stabilizersym}, there exists a quantum
state $\ket{\psi}$ in the stabilizer code $\mathcal{S}$ which also has
the symmetry $\sigma$.  Hence $\ket{\psi}$ is the stabilizer state
with the desired symmetry. Note that the stabilizer $\mathcal{S}'$ of
the state $\ket{\psi}$ contains the original stabilizer
$\mathcal{S}$. Therefore, common eigenspaces of $\mathcal{S}$ are
further decomposed into one-dimensional joint eigenspaces of
$\mathcal{S}'$, and we can rewrite the projection $P_{\mathcal{Q}}$
onto the USt code in the form corresponding to a CWS
code.
\end{proof}

\section{Symmetries of the Classical Code}

Theorem~\ref{thm:CWSsym} does not make any statement about the
symmetry of the classical code.  In general, if we insist to use the
canonical form of the CWS code as given by Theorem~\ref{thm:unique}, we
cannot expect that the (non-linear) classical code $\mathcal{C}$
associated with the CWS code $\mathcal{Q}=(\mathcal{S},\mathcal{C})$
has the same symmetry as $\mathcal{Q}$.  That is, in this case, even
if the stabilizer $\mathcal{S}$ has the same permutation symmetry
$\sigma$ as the quantum code $\mathcal{Q}$, one might not be able to
find a classical code $\mathcal{C}$ with the same symmetry $\sigma$ in
general. Let us look at an example.

\begin{example}
Consider the stabilizer state $1/\sqrt{2}(\ket{00\ldots
  0}-\ket{11\ldots 1})$ (hence a CWS code, denoted by $\mathcal{Q}$),
which is invariant under all permutations.  Using the canonical form
of $\mathcal{Q}=(\mathcal{S},\mathcal{C})$ as given by
Theorem~\ref{thm:unique}, the group $\mathcal{S}$ is generated by
$XX\dots X$ and all pairs of $Z$, which is permutation
invariant.  However, the classical code $\mathcal{C}$ consists of the
vector which is one in the first coordinate and zero elsewhere, i.\,e.,
$\mathcal{C}$ is a code with a single codeword $10\ldots0$, which has
a smaller symmetry group than that of $\mathcal{Q}$.

On the other hand, if we choose the group $\mathcal{S}'$ generated
by $-XX\dots X$ and all pairs of $Z$, the corresponding classical code
$\mathcal{C}'$ consists just of the zero vector.  So in the
representation $\mathcal{Q}=(\mathcal{S}',\mathcal{C}')$,  both
$\mathcal{S}'$ and $\mathcal{C}'$ have the same permutation symmetries
as $\mathcal{Q}$.
\end{example}

This example indicates that exploiting the phase factor freedom in the
USt code decomposition of a CWS code, and thereby deviating slightly
from the canonical form, there is some chance to find both a
stabilizer and a classical code with the same permutation symmetry as
the CWS code.

To study the properties of the classical code $\mathcal{C}$ associated
with a CWS code $\mathcal{Q}=(\mathcal{S},\mathcal{C})$, consider the
case where the stabilizer state $\mathcal{S}$ has some permutation
symmetry $\sigma$.  Then for given generators
$\{\bm{g}_1,\bm{g}_2,\ldots,\bm{g}_n\}$ of the stabilizer
$\mathcal{S}$, the permuted operators
$\{\bm{g}_1^\sigma,\bm{g}_2^\sigma,\ldots,\bm{g}_n^\sigma\}$ generate
the same stabilizer $\mathcal{S}$.  The transformation
$\bm{g}_i\mapsto \bm{g}_i^\sigma$ can be characterized by a
$\mathbb{Z}_p$-valued, invertible $n\times n$ matrix $R$ given by
\begin{equation}
\label{eq:gentransR}
\bm{g}_i^\sigma=\prod\limits_{j=1}^n \bm{g}_j^{{R}_{ji}}.
\end{equation}

Let us write the $K$ classical codewords in $\mathcal{C}$ as an
$K\times n$ matrix with entries $c_{ij}$.  We are now ready to present
the following theorem, which gives a sufficient condition for
$\mathcal{C}$ to guarantee that $\mathcal{Q}$ has the same permutation
symmetry as $\mathcal{S}$
\begin{theorem}
\label{thm:suff}
Let $\mathcal{Q}=(\mathcal{S},\mathcal{C})$ be a CWS code, and let
$\{\bm{g}_1,\bm{g}_2,\ldots,\bm{g}_n\}$ be generators of
$\mathcal{S}$. If $\mathcal{S}$ has permutation symmetry $\sigma$,
where $\bm{g}_i^\sigma=\prod\limits_{j=1}^n \bm{g}_j^{{R}_{ji}}$, and
$\mathcal{C}{R}\cong\mathcal{C}$, then $\mathcal{Q}$ has the same
permutation symmetry $\sigma$ as $\mathcal{S}$. Here by
$\mathcal{C}R\cong\mathcal{C}$ we mean that the set of rows of
$\mathcal{C}R$, corresponding to the transformed code, equals the code
$\mathcal{C}$ (not as a matrix).
\end{theorem}

\begin{proof}
We start by applying the permutation $\sigma$ to the projection
$P_{\mathcal{Q}}$  onto the code space
given by Eq.~\eqref{eq:CWSproj}:
\begin{alignat}{7}
\label{eq:CWStrans}
P_{\mathcal Q}^\sigma
&{}=\sum_{\bm{t}\in{\mathcal{C}}}\frac{1}{p^n}\sum_{\bm{x}\in\F_p^n} \omega^{\bm{t}\cdot\bm{x}}(\bm{g}_1^\sigma)^{x_1}\dots(\bm{g}_n^\sigma)^{x_n}
=\sum_{\bm{t}\in{\mathcal{C}}}\frac{1}{p^n}\sum_{\bm{x}\in\F_p^n} \omega^{\bm{t}\cdot\bm{x}}(\prod_{j}\bm{g}_j^{R_{j1}})^{x_1}
\dots(\prod_j\bm{g}_j^{R_{jn}})^{x_n}\nonumber\\
&{}=\sum_{\bm{t}\in{\mathcal{C}}}\frac{1}{p^n}\sum_{\bm{x}\in\F_p^n} \omega^{\bm{t}\cdot\bm{x}}\bm{g}_1^{\sum_jR_{1j}x_j}
\dots\bm{g}_n^{\sum_jR_{nj}x_j}
\end{alignat}
Let $x_j'=\sum_iR_{ji}x_i$ and $t_i=\sum_jR_{ji}t'_j$.  Then for
$\bm{t}\in\mathcal{C}$, we have $\bm{t}'\in\mathcal{C}'$, where the
transformed code $\mathcal{C}'$, considered as a $K\times n$ matrix,
is given by
\begin{equation}
\label{eq:ctrans}
\mathcal{C}=\mathcal{C}'R.
\end{equation}
Then Eq.~\eqref{eq:CWStrans} becomes
\begin{eqnarray}
\label{eq:CWStransC}
P_{\mathcal Q}^\sigma
&=&\sum_{\bm{t}'\in{\mathcal{C}'}}\frac{1}{p^n}\sum_{\bm{x}\in\F_p^n} \omega^{\sum_i\sum_jR_{ji}t'_jx_i}{\bm{g}}_1^{\sum_jR_{1j}x_j}
\dots{\bm{g}}_n^{\sum_jR_{nj}x_j}\nonumber\\
&=&\sum_{\bm{t}'\in\mathcal{C}'}\frac{1}{p^n}\sum_{\bm{x}\in\F_p^n} \omega^{\sum_jt'_j(\sum_iR_{ji}x_i)}{\bm{g}}_1^{\sum_jR_{1j}x_j}
\dots{\bm{g}}_n^{\sum_jR_{nj}x_j}\nonumber\\
&=&\sum_{\bm{t}'\in\mathcal{C}'}\frac{1}{p^n}\sum_{\bm{x}'\in\F_p^n} \omega^{\bm{t}'\cdot\bm{x}'}{\bm{g}}_1^{x_1'}
\dots{\bm{g}}_n^{x_n'}.
\end{eqnarray}
Now because of $\mathcal{C}R\cong\mathcal{C}$, the rows of
$\mathcal{C}R$ are a permutation of the rows of
$\mathcal{C}$.  Hence there exists a permutation matrix $P$ such that
$P\mathcal{C}R=\mathcal{C}$, which gives
\begin{equation}
P\mathcal{C}=\mathcal{C}R^{-1}=\mathcal{C}'.
\end{equation}
The second equality follows from Eq.~\eqref{eq:ctrans}. Hence the rows
of $\mathcal{C}'$ are a permutation of the rows of $\mathcal{C}$,
i.\,e., $\mathcal{C}$ and $\mathcal{C}'$ are the same code. Therefore
Eq.~\eqref{eq:CWStransC} becomes
\begin{eqnarray}
\label{eq:CWStransC1}
P_{\mathcal Q}^\sigma
=\sum_{\bm{t}'\in\mathcal{C}}\frac{1}{p^n}\sum_{\bm{x}'\in\F_p^n} \omega^{\bm{t}'\cdot\bm{x}'}{\bm{g}}_1^{x_1'}
\dots{\bm{g}}_n^{x_n'}= P_{\mathcal Q},
\end{eqnarray}
which proves the theorem.
\end{proof}

Note that although Theorem~\ref{thm:suff} is stated in terms of a set of
generator $\bm{g}_i$ of $\mathcal{S}$, it is actually independent of
the choice of the generators. That is to say, if
$\bm{g}_i^\sigma=\prod\limits_{j=1}^n \bm{g}_j^{R_{ji}}$, and
$\mathcal{C}R\cong\mathcal{C}$ holds, then for some other generators
$\bm{g}_i'$ of $\mathcal{S}$, where
$\mathcal{Q}=(\mathcal{S}_{\bm{g}'},\mathcal{C}')$ and
$(\bm{g}'_i)^\sigma=\prod\limits_{j=1}^n (\bm{g}'_j)^{R'_{ji}}$, one
would then have $\mathcal{C}'R'\cong\mathcal{C}'$.

Theorem~\ref{thm:suff} gives a sufficient condition that the CWS code
$\mathcal{Q}=(\mathcal{S},\mathcal{C})$ may have the same permutation
symmetry as $\mathcal{S}$. Note that \cite[Proposition 5.2]{DK10}
considers a special case of Proposition~\ref{thm:suff}, where the
permutation $\sigma$ is the cyclic shift.  However, it turns out that
the argument in \cite{DK10} is false; the cyclic symmetry of the
stabilizer $\mathcal{S}$ is not sufficient to guarantee the cyclic
symmetry of the resulting quantum code $\mathcal{Q}$; the classical
code $\mathcal{C}$ must also have a cyclic symmetry, as discussed in
Corollary~\ref{pro:cyclicgraph}.

It remains unclear whether the condition given in
Theorem~\ref{thm:suff} is also necessary, at least in the case when
both the CWS code $\mathcal{Q}$ and the stabilizer $\mathcal{S}$ have
a permutation symmetry $\sigma$.  We expect that in this case the
condition $\mathcal{C}R\cong\mathcal{C}$ would be necessary. However,
while the condition might be violated for a particular choice of
$\mathcal{S}$, it might hold for a different representation
$\mathcal{Q}=(\mathcal{S}',\mathcal{C}')$.

\section{The Standard Form}

Starting with the unique representation of a CWS code as a USt code,
we can derive a standard form of a CWS code.  We know that up to local
Clifford (LC) operations, any CWS code $\cal{Q}$ can be represented by
a graph $\cal{G}$ and a binary classical code $\cal{C}$
\cite{CCS+09,CSS+09}.  Starting with a given CWS code
$\mathcal{Q}=(\mathcal{S},\mathcal{C})$, one can transform the
stabilizer $\cal{S}$ into a graph state using LC operations, and then
$\mathcal{C}$ will be transformed accordingly \cite{CCS+09}.  Our
concern is that if $\cal{Q}$ has some permutation symmetry $\sigma$,
whether it can be kept during this LC operations, in other words,
whether one can always obtain a graph with the same permutation
symmetry $\sigma$ as $\cal{Q}$ has.

Indeed, even if one can always find a stabilizer state $\mathcal{S}$
with the same symmetry as $\cal{Q}$ has, we are asking too much here
for the graph $\cal{G}$.  In general, one cannot find a graph with the same 
permutation symmetry as $\cal{Q}$ has. Let us look at an example.
\begin{example}
The stabilizer $\mathcal{S}$ for the $7$-qubit Steane code is generated by
\begin{equation}
{\bm g}_1=XIXXXII,\quad {\bm g}_2=IXIXXXI,\quad {\bm g}_3=IIXIXXX,
\end{equation}
which are the three $X$-type generators, and the three $Z$-type
generators
\begin{equation}
{\bm g}_4=ZIZZZII,\quad {\bm g}_5=IZIZZZI,\quad {\bm g}_6=IIZIZZZ.
\end{equation}
This code is cyclic, and for its CWS representation, one can choose,
e.\,g., the stabilizer state $\ket{\psi}$ with stabilizer
$\mathcal{S}'$ generated by
$\mathcal{S}'=\langle\bm{g}_1,\bm{g}_2,\bm{g}_3,\bm{g}_4,\bm{g}_5,\bm{g}_6,Z^{\otimes
  7}\rangle$.  Then $\ket{\psi}$ is cyclic as well.  However, when
transforming the Steane code into the standard form of its CWS
representation, one cannot find it a cyclic graph \cite{GKR02}. In
fact, the best symmetric graph one can find is with a three-fold
cyclic symmetry instead of a $7$-fold cyclic symmetry, as shown in
Fig.~\ref{fig1}(a).  The three-fold symmetry is in fact the symmetry
of the generators of $\mathcal{S}'$ instead of the symmetry of the
entire stabilizer group $\mathcal{S}'$.  This is related to the fact
that the graph $\cal{G}$ in some sense represents only the stabilizer
generators of its corresponding graph state.
\end{example}

\begin{figure}[htbp]
\vskip-1ex
\centerline{\includegraphics[width=2.4in,angle=0]{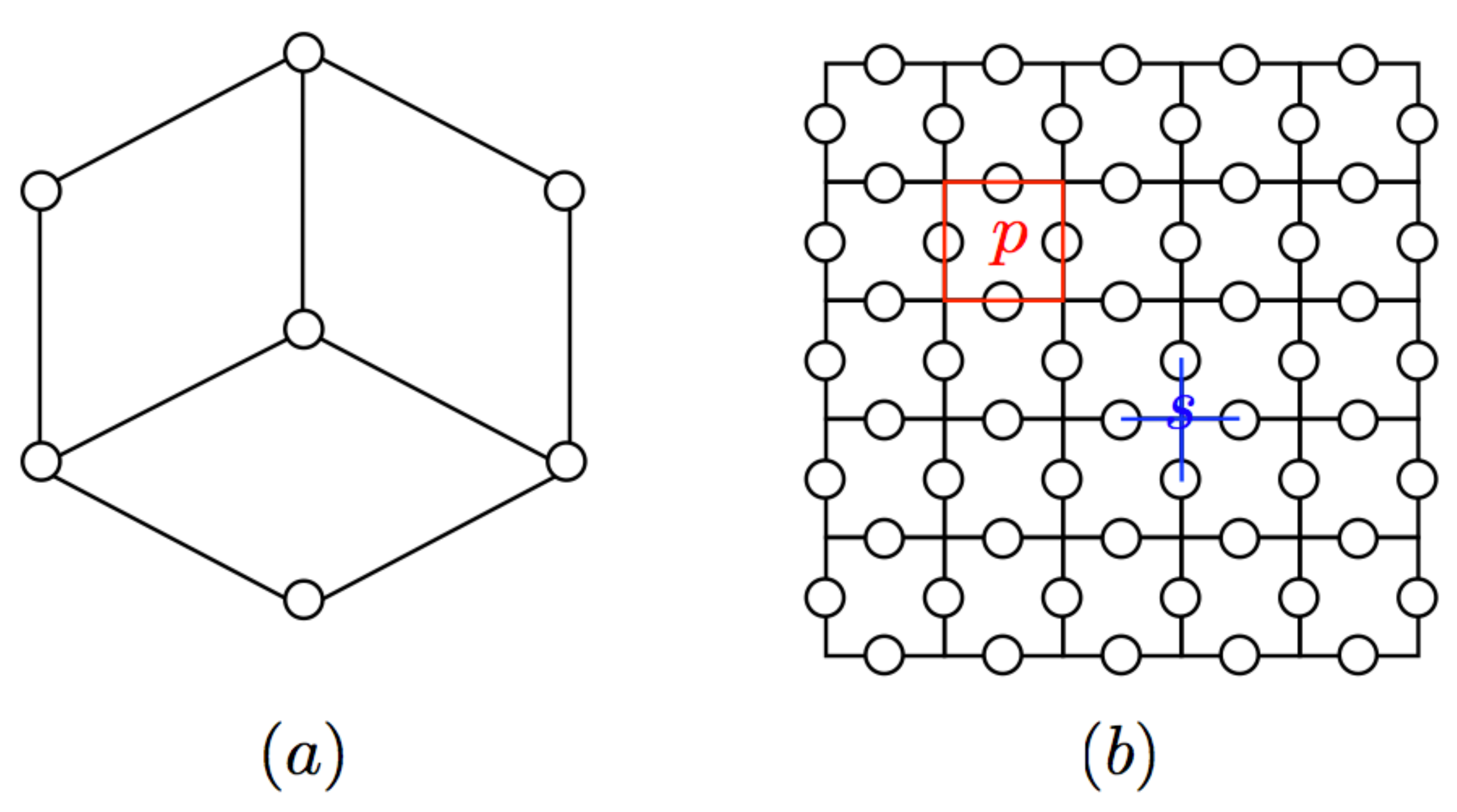}}
\vskip-1ex
\caption{(a) The graph for the Steane code with three-fold cyclic
  symmetry. (b) The toric code on a square lattice. Qubits are sitting
  on edges of the lattice. $p$ denotes a plaquette, which contains $4$
  qubits as shown across the red lines. $s$ denotes a star, which
  contains $4$ qubits as shown across the blue lines.\label{fig1} }
\end{figure}

The toric code turns out to provide another example, as shown in
Fig.~\ref{fig1}(b), which is in some sense even worse than the
Steane code example. Despite the fact that the generators of the
stabilizer group for the toric code have a translational symmetry, we
will show in Theorem~\ref{thm:toric} that one cannot find a graph with
translational symmetry. However, both the Steane code and the toric
code do not provide counterexamples to Theorem~\ref{thm:CWSsym}, as
the logical zero has the desired symmetry in both cases.

Nevertheless, there might still be some interesting relationship
between the permutation symmetries of $\cal{Q}$ and the symmetries of
$\cal{G}$ and $\cal{C}$. Let us start with a simple case:
\begin{corollary}
\label{pro:cyclicgraph}
For a CWS code $\cal{Q}=(\cal{G},\cal{C})$, if both $\cal{G}$ and
$\cal{C}$ have a permutation symmetry $\sigma$, then the code
$\mathcal{Q}$ has the permutation symmetry $\sigma$ as well.
\end{corollary}

\begin{proof}
This is actually a direct implication of Theorem~\ref{thm:suff}; in
this case the matrix $R$ is nothing but a permutation matrix
corresponding to the permutation $\sigma$.
\end{proof}

This turns out to be good luck, as due to the structure of the
stabilizer generators of graph states, a permutation of the qubits
corresponds to the same permutation of the generators $\bm{g}_i$, and
hence also corresponds to a permutation of the coordinates in the
classical code $\mathcal{C}$.  Prominent examples are the
$(\!(5,2,3)\!)$ code and the $(\!(5,6,2)\!)$ code, whose corresponding
graph is a pentagon in both cases, and the corresponding classical
codes are cyclic (see \cite[Sec. IIIA,B]{CSS+09}).

Finally, let us examine the graph symmetry for the toric code. The
toric code was first proposed by Kitaev in 1997 as an example
demonstrating topologically ordered quantum systems
\cite{Kit97,KSV:computation}.  The setting is a two-dimensional square
lattice with periodic boundary conditions and with a qubit sitting on
each edge of the lattice. There are two types of stabilizer
generators:
\begin{enumerate}
\item \figstar\ (star) type, indicated in Fig.~\ref{fig1}(b) as $s$: 
\begin{equation}
\label{eq:star}
A_s^{X}=\prod_{j\in\text{star}(s)}X_j
\end{equation}
\item \figbox\ (plaquette) type, indicated in Fig.~\ref{fig1}(b) as $p$:
\begin{equation}
\label{eq:plaquette}
A_p^{Z}=\prod_{j\in\text{plaquette}(p)}Z_j
\end{equation}
\end{enumerate}
It is straightforward to check that $A_s^{X}$ and $A_p^Z$ commute for
any pair $s$ and $p$. 

These stabilizer generators are by definition translational invariant,
for the translation along each direction of the two-dimensional square
lattice.  What is more, one can even
find an encoding graph which is also translational invariant
\cite{BR07}.  We will show that unfortunately one cannot find a
translational invariant graph to represent the toric code as a CWS
code.

\begin{theorem}
\label{thm:toric}
A graph corresponding to the toric code cannot have the same translational symmetry
as the code.
\end{theorem}
\begin{proof}
Let $\mathcal T$ be the toric code stabilizer generated by the star
and plaquette operators as given by Eq. (\ref{eq:star}) and
Eq. (\ref{eq:plaquette}).  Suppose that $\mathcal Q=(\mathcal{G},
\mathcal C)$ is a code where $\mathcal G$ is as symmetric as the toric
code stabilizer (i.\,e., translational invariant) and is local
Clifford equivalent to $\mathcal T$. This means that if we let
$\mathcal S$ be the stabilizer of $\mathcal Q$, then there are local
Clifford elements $C_1, C_2, \dots$ such that $\mathcal S
=\widehat{C}\mathcal {T}\widehat C^{\dagger}$, where
$\widehat{C}=C_1\otimes C_2\otimes \cdots$ (here we choose an
arbitrary indexing of qubits).

Let $\sigma$ be a permutation symmetry of the toric code and define
$\widehat{C}_{\sigma}=C_{\sigma(1)}\otimes C_{\sigma(2)}\otimes
\cdots$. Since $\sigma$ is assumed to be a symmetry of $\mathcal S$ as
well, we have
$$
\mathcal S
=\widehat{C}\mathcal {T}\widehat
C^{\dagger}=\widehat{C}_{\sigma}\mathcal {T}\widehat
C^{\dagger}_{\sigma}.
$$
Then for $D_i=C_i^{\dagger}C_{\sigma(i)}$, we have $\widehat{D}\mathcal{T}
\widehat D^{\dagger}=\mathcal{T}$, where $\widehat{D}=D_1\otimes D_2\otimes \cdots$.

Let $XXXX$ be the element of
this stabilizer group $\mathcal T$ corresponding to some star
\figstar. Since $\widehat{D}$ is local, and $XXXX$ is the
only element of $\mathcal T$ that acts on edges
corresponding to \figstar, we must have
$\mathcal{D}XXXX\mathcal{D}^{\dagger}=XXXX$. The same argument applies
to the $Z$-terms corresponding to a plaquette \figbox\,. As a result,
conjugation by $D_i$ maps $X$ to $\pm X$ and $Z$ to $\pm Z$. Hence
$D_i$ is an element of the Pauli group.

Now we know that $\widehat{D}$ is in the Pauli group, and it holds for
every permutation $\sigma$. On the other hand, the symmetry group of
the toric code is transitive. Therefore, for every $i$, $j$, the product
$C_i^{\dagger}C_j$ is in the Pauli group, and furthermore
\[
C_1\otimes C_2\otimes \cdots = 
(H\otimes H\otimes \cdots ) (P_1\otimes P_2\otimes \cdots),
\]
where the factors $P_i$ are in the Pauli group and $H$ is some
Clifford element acting on a single qubit. 

$\widehat C \mathcal T \widehat C^\dagger$ is supposed to correspond to a graph state, but $(P_1\otimes P_2\otimes
\cdots)$ just changes some signs in the stabilizer group, and 
$(H\otimes H\otimes \cdots )$ cannot turn the stabilizer group of the
toric code into a graph-type stabilizer group. 
\end{proof}

\section{Summary and Discussion}

In this work we have investigated the symmetry properties of CWS
codes.  Our main result shows that for a given CWS code $\mathcal{Q}$
with some permutation symmetry $\sigma$, there always exits a
stabilizer state $\mathcal{S}$ with the same symmetry $\sigma$ such
that $\mathcal{Q}=(\mathcal{S},\mathcal{C})$ for some classical code
$\mathcal{C}$. As many good CWS codes are found by starting from a
chosen $\mathcal{S}$, this ensures that when trying to find CWS codes
with certain permutation symmetry, the choice of $\mathcal{S}$ with
the same symmetry will suffice. A key point to reach our main result
is to obtain a canonical representation for CWS codes, i.\,e., a
unique decomposition as USt codes.

One natural question is whether there is any chance to find a
classical code $\mathcal{C}$ with the same symmetry $\sigma$ as that
of $\mathcal{Q}$, which, together with some $\mathcal{S}$ with
symmetry $\sigma$, gives $\mathcal{Q}=(\mathcal{S},\mathcal{C})$. We
do not know the answer in general, but we know that one can no longer
restrict $\mathcal{S}$ to the stabilizer used in the canonical form,
but might have to introduce some phase factors.  We have developed a
sufficient condition that $\mathcal{C}$ has to satisfy in order to
ensure that in combination with some $\mathcal{S}$ with symmetry
$\sigma$, one will have $\mathcal{Q}=(\mathcal{S},\mathcal{C})$ with
the same symmetry $\sigma$.  Observing the fact that the permutation
on the code $\mathcal{Q}$ does not directly translate into a
permutation of the classical $\mathcal{C}$ (but a linear
transformation given by the matrix $R$), in general one cannot expect
to find a classical code $\mathcal{C}$ with the same symmetry as that
of $\mathcal{Q}$.

One interesting case are cyclic codes.  If there exists a graph
$\mathcal{G}$ which has the same symmetry $\sigma$ as the CWS code
$\mathcal{Q}=(\mathcal{G},\mathcal{C})$, then the permutation of the
code $\mathcal{Q}$ translates directly into a permutation of the
classical code $\mathcal{C}$. Hence, combining a graph $\mathcal{G}$
whose symmetry group contains the cyclic group of order $n$, with a
cyclic classical code $\mathcal{C}$ of length $n$, gives a cyclic CWS
code $\mathcal{Q}=(\mathcal{G},\mathcal{C})$.  It would be nice to see
whether the converse is true as well, i.\,e., given a cyclic CWS code
$\mathcal{Q}$ which corresponds to a graph $\mathcal{G}$ whose
symmetry group contains the cyclic group of order $n$, can we always
find a cyclic classical code $\mathcal{C}$ of length $n$, such that
$\mathcal{Q}=(\mathcal{G},\mathcal{C})$. We leave this for future
investigation.

In general, although every CWS code $\mathcal{Q}$ is local Clifford
equivalent to a standard form $(\mathcal{G},\mathcal{C})$, the local
Clifford operation may destroy the permutation symmetry of the
original code.  In other words, one cannot expect to always find a
graph $\mathcal{G}$ which has the same symmetry as that of
$\mathcal{Q}$.  The seven-qubit Steane code is such an example where
the graph can only possess a three-fold cyclic symmetry which is the
symmetry of the stabilizer generators, instead of the seven-fold
cyclic symmetry of the code.  For the toric code, despite the stabilizer
generators being translational invariant, we show that there does not
exist any associated translational invariant graph.  A general
understanding of the conditions that graphs can possess the same
symmetry as the CWS code is worth further investigation.

\subparagraph*{Acknowledgements}

SB was in part supported by National Elites Foundation and by a grant
from IPM (No. 91810409). JC is supported by NSERC and NSF of China
(Grant No. 61179030).  The CQT is funded by the Singapore MoE and the
NRF as part of the Research Centres of Excellence programme.  ZJ
acknowledges support from NSERC, ARO and NSF of China (Grant
Nos. 60736011 and 60721061).  QW is supported by NSERC.  BZ is
supported by NSERC and CIFAR.
MG acknowledges support by the Intelligence Advanced Research Projects
Activity (IARPA) via Department of Interior National Business Center
contract No. D11PC20166.  The U.S. Government is authorized to reproduce
and distribute reprints for Governmental purposes notwithstanding any
copyright annotation thereon. Disclaimer: The views and conclusions
contained herein are those of the authors and should not be interpreted
as necessarily representing the official policies or endorsements,
either expressed or implied, of IARPA, DoI/NBC, or the
U.S. Government.

The authors would like to thank Martin R{\"o}tteler for his suggestion
to use the Fourier transformation to prove Lemma \ref{lm:shift}.

\bibliography{CWSSym}







\end{document}